\documentclass[12pt]{article} % flatex input: [prefix.tex]

\usepackage[bibencoding=utf8,style=alphabetic,backend=biber]{biblatex}%
\usepackage{sariel_biblatex}%

\usepackage{amsmath}%
\usepackage{amssymb}%
\usepackage[table]{xcolor}%

\usepackage[amsmath,thmmarks]{ntheorem}%
\usepackage{microtype}

\usepackage{fix-cm}
\usepackage[T1]{fontenc}

\usepackage{xcolor}%
\usepackage{mleftright}%
\usepackage{xspace}%
\usepackage{graphicx}
\usepackage{hyperref}%
\usepackage[inline]{enumitem} %
\usepackage{hyperref}%
\usepackage[ocgcolorlinks]{ocgx2} %

\hypersetup{%
      unicode,
      breaklinks,%
      colorlinks=true,%
      urlcolor=[rgb]{0.25,0.0,0.0},%
      linkcolor=[rgb]{0.5,0.0,0.0},%
      citecolor=[rgb]{0,0.2,0.445},%
      filecolor=[rgb]{0,0,0.4},
      anchorcolor=[rgb]={0.0,0.1,0.2}%
}

\usepackage{titlesec}%
\titlelabel{\thetitle. }%

\usepackage[amsmath,thmmarks]{ntheorem}%
\usepackage{pgffor} %

\theoremseparator{.}%

\theoremstyle{plain}%
\theoremheaderfont{\normalfont\bfseries}
\theorembodyfont{\itshape}
\newtheorem{theorem}{Theorem}[section]

\foreach \env/\display in {
   lemma/Lemma, conjecture/Conjecture, corollary/Corollary, claim/Claim,  invariant/Invariant,
   proposition/Proposition, prop/Proposition, openproblem/Open Problem%
}{ \edef\temp{{\noexpand\newtheorem{\env}[theorem]{\display}}} \temp }

\theoremseparator{}%

\theoremseparator{.}%

\definecolor{darkslate}{rgb}{0.2, 0.2, 0.2} %

\theoremseparator{.}%
\theoremheaderfont{\sffamily\bfseries\color{darkslate}}
\theorembodyfont{\upshape}%
\foreach \env/\display in {%
   defn/Definition, definition/Definition, example/Example, exercise/Exercise, xca/Exercise, problem/Problem,
   question/Question,assumption/Assumption%
}{\edef\temp{{\noexpand\newtheorem{\env}[theorem]{\display}}} \temp}

\theoremstyle{plain}%
\theoremseparator{.}%
\theoremheaderfont{\sffamily}%
\theorembodyfont{\upshape}%
\newtheorem*{remark:unnumbered}[theorem]{Remark}%

\foreach \env/\display in {
   remark/Remark, note/Note, fact/Fact, observation/Observation
}{ \edef\temp{{\noexpand\newtheorem{\env}[theorem]{\display}}} \temp }

\usepackage{qed_duck}

\theoremheaderfont{\itshape}%
\theorembodyfont{\upshape}%
\theoremstyle{nonumberplain}%
\theoremseparator{}%
\theoremsymbol{\myqedsymbol}%
\newtheorem{proof}{Proof:}%

\definecolor{blue25emph}{rgb}{0, 0, 11}
\definecolor{almostblack}{rgb}{0, 0, 0.3}

\providecommand{\emphi}[1]{}
\renewcommand{\emphi}[1]{\textcolor{blue25emph}{\textbf{\emph{#1}}}}

\providecommand{\emphw}[1]{}%
\renewcommand{\emphw}[1]{{\textcolor{almostblack}{\emph{#1}}}}%

\providecommand{\emphOnly}[1]{}%
\renewcommand{\emphOnly}[1]{\emph{\textcolor{blue25emph}{\textbf{#1}}}}

\newcommand{\SarielThanks}[1]{%
   \thanks{%
      School of Computing and Data Science; %
      University of Illinois; %
      201 N. Goodwin Avenue; %
      Urbana, IL, 61801, USA; %
      \href{mailto:spam@illinois.edu}{sariel@illinois.edu}; %
      \url{http://sarielhp.org/}.%
   #1%
   }%
}

\newcommand{\HLink}[2]{\hyperref[#2]{#1~\ref*{#2}}}
\newcommand{\HLinkSuffix}[3]{\hyperref[#2]{#1\ref*{#2}{#3}}}

\newcommand{\thmlab}[1]{{\label{theo:#1}}}
\newcommand{\thmref}[1]{\HLink{Theorem}{theo:#1}}

\newcommand{\lemlab}[1]{\label{lemma:#1}}
\newcommand{\lemref}[1]{\HLink{Lemma}{lemma:#1}}%

\newcommand{\seclab}[1]{\label{sec:#1}}
\newcommand{\secref}[1]{\HLink{Section}{sec:#1}}

\newcommand{\defrefY}[2]{\hyperref[def:#1]{#2}}

\providecommand{\eqlab}[1]{}%
\renewcommand{\eqlab}[1]{\label{equation:#1}}
\newcommand{\Eqref}[1]{\HLinkSuffix{Eq.~(}{equation:#1}{)}}
\newcommand{\eqrefY}[2]{\hyperref[equation:#2]{#1}}
\newcommand{\dirEdgeY}[2]{\pth{#1, #2}}%

\newcommand{\Set}[2]{\left\{ #1 \;\middle\vert\; #2 \right\}}

\newcommand{\pth}[1]{\mleft(#1\mright)}%

\newcommand{\cardin}[1]{\left\lvert {#1} \right\rvert}%

\newlist{compactenumA}{enumerate}{5}%
\setlist[compactenumA]{itemsep=-0.5ex,topsep=0.5ex,partopsep=1ex,parsep=1ex,%
   label=(\Alph*)}%

\newlist{compactenuma}{enumerate}{5}%
\setlist[compactenuma]{itemsep=-0.5ex,topsep=0.5ex,partopsep=1ex,parsep=1ex,%
   label=(\alph*)}%

\newlist{compactenumI}{enumerate}{5}%
\setlist[compactenumI]{itemsep=-0.5ex,topsep=0.5ex,partopsep=1ex,parsep=1ex,%
   label=(\Roman*)}%

\newlist{compactenumi}{enumerate}{5}%
\setlist[compactenumi]{itemsep=-0.5ex,topsep=0.5ex,partopsep=1ex,parsep=1ex,%
   label=(\roman*)}%

\newlist{compactitem}{itemize}{5}%
\setlist[compactitem]{itemsep=-0.5ex,topsep=0.5ex,partopsep=1ex,parsep=1ex,%
   label=\ensuremath{\bullet}}%

\newcommand{\etal}{\textit{et~al.}\xspace}

\numberwithin{figure}{section}%
\numberwithin{table}{section}%
\numberwithin{equation}{section}%

\newcommand{\LPVC}{(\Term{LPVC})\xspace}%
\newcommand{\xx}{\mathsf{x}}%
\newcommand{\yy}{\mathsf{y}}%
\newcommand{\zz}{\mathsf{z}}%
\newcommand{\Vertices}{V}%
\newcommand{\VerticesX}[1]{\Vertices\pth{#1}}
\newcommand{\Edges}{E}%
\newcommand{\EdgesX}[1]{\Edges\pth{#1}}
\newcommand{\VX}[1]{\VerticesX{#1}}%

\newcommand{\alg}{\texttt{algIS}\xspace}%
\newcommand{\Gr}{G}%
\newcommand{\GrA}{H}%
\newcommand{\GrB}{K}%
\newcommand{\GrC}{N}%
\newcommand{\eps}{{\varepsilon}}%
\newcommand{\ProblemC}[1]{{\textsf{{#1}\index{problem!#1}}}}

\newcommand{\Term}[1]{\textsf{#1}}
\newcommand{\VC}{\Term{VC}\xspace}%

\newcommand{\vcOpt}{\mathrm{vc}^{*}}
\newcommand{\isOpt}{\mathrm{is}^{*}}%
\DefineNamedColor{named}{OliveGreen} {rgb}{0.0,0.1,0.0}
\providecommand{\ComplexityClass}[1]{{{\textcolor[named]{OliveGreen}{%
            \textsc{\textbf{#1}}}}}}
\providecommand{\NPHard}{{\ComplexityClass{NP-Hard}}\xspace}
\providecommand{\NP}{\ComplexityClass{NP}\xspace}
\providecommand{\POLYT}{\ComplexityClass{P}\xspace}

\newcommand{\MWVC}{\Term{MWVC}\xspace}%
\newcommand{\LP}{\Term{LP}\xspace}%

\newcommand{\PTAS}{\Term{PTAS}\xspace}%
\newcommand{\QPTAS}{\Term{QPTAS}\xspace}%
\newcommand{\QQPTAS}{\Term{QQPTAS}\xspace}%

\newcommand{\Sopt}{S_{\mathrm{opt}}}

\newcommand{\valX}[1]{\nu\pth{#1}}
\newcommand{\poly}{\mathrm{poly}}

\bibliography{indep_set_to_vc}%

\begin{document}

\title{Approximately: Independence Implies Vertex Cover}

\author{Sariel Har-Peled\SarielThanks{Work on this paper was partially supported by an NSF AF award CCF-1907400.  %
   }}

\date{\today}

\maketitle

\begin{abstract}
    We observe that a $(1-\eps)$-approximation algorithm to \ProblemC{Independent Set}, that works for any induced subgraph of the input graph, can be used, via a polynomial time reduction, to provide a $(1+\eps)$-approximation to \ProblemC{Vertex Cover}. This basic observation was made before, see \cite{bhr-mvcrg-11}.

    As a consequence, we get a \PTAS for \VC for unweighted pseudo-disks, \QQPTAS\footnote{We refer to an $(1\pm \eps)$-approximation algorithm with running time of the form $n^{O( \poly(\log \log n, 1/\eps))}$ as \QQPTAS.  The more traditional running time $n^{O( \poly(\log n, 1/\eps))}$ is \QPTAS (i.e.,quasi-polynomial time approximation scheme). A \PTAS refers to an approximation algorithm with running time of the form $n^{O( \poly( 1/\eps))}$. Further mysterious acronyms can be provided upon request.} for \VC for unweighted axis-aligned rectangles in the plane, and \QPTAS for \MWVC for weighted polygons in the plane. To the best of our knowledge, all these results are new.
\end{abstract}

\section{Introduction}

\paragraph{Background.}
Given a graph $\Gr=(\Vertices, \Edges)$, a \emphi{vertex cover} (\VC) is a set $C \subseteq \Vertices$, that is adjacent to all the edges of $\Gr$. The problem of computing a minimum \VC is a classical problem that is \NPHard \cite{gj-cigtn-79}, and an easy $2$-approximation algorithm is known (by computing greedily a maximal matching, and using its vertices). Dinur and Safra \cite{s-ohavc-05} showed that no approximation better than $1.3606$ is possible, unless $\POLYT = \NP$. This was later improved to $\sqrt{2} - \eps \approx 1.41421$ by Khot \etal \cite{kms-psggn-18}.  If the unique-game conjecture is true, no approximate better than $2$ is possible in the general case \cite{kr-vcmha-08}.

For a geometric intersection graph, the problem is easier. Erlebach \etal \cite{ejs-ptasg-05} gave a \PTAS for the intersection graph of weighted fat objects. Har-Peled and Quanrud \cite{hq-aapel-17} provided a \PTAS for \VC for the unweighted case for low-density graphs, or polynomial expansion graphs.

\paragraph{The observation.}
Given an instance of \ProblemC{Vertex Cover}, one can reduce it in polynomial time into a dense instance, where the \VC is at least half the vertices in the graph (or half the mass in the weighted case). In such a graph, a $(1-\eps)$-approximation to the independent set implies a $(1+\eps)$-approximation readily to \VC.

\paragraph{Outline.} %
For the sake of completeness, and to provide a relatively self-contained description, we review this reduction (into the dense subgraph), in excruciating detail, in \secref{background}. We emphasize, however, that this reduction is well known -- see \cite{cc-intti-04,cc-crmwv-08} and references therein. The unweighted case is somewhat easier, and is described nicely in Cygan \etal \cite[Chapter 2]{cfklm-pa-15}.  We describe the new result and its applications in \secref{result}.

\section{Reduction of \VC to the dense case}
\seclab{background}

The material covered in this section can be found in \cite{cfklm-pa-15, cfklm-pa-15,cc-intti-04,cc-crmwv-08} and references therein.

Given a graph $\Gr=(\Vertices,\Edges)$, and a weight function $w: \Vertices \rightarrow (0,\infty)$, the \emphi{min weight vertex cover} (\MWVC) is the problem of computing the subset $\Sopt \subseteq \Vertices$, such that $\Sopt$ is adjacent to all edges in $\Gr$, and $\vcOpt(\Gr) = w(\Sopt) = \sum_{v \in \Sopt} w(v)$ is the minimum among all such sets.  The complement of a vertex cover is an independent set, and in particular, we have $\vcOpt(\Gr) = w(\Gr) - \isOpt(\Gr)$, where $\isOpt$ is the weight of the maximum weight independent set in $\Gr$, and $w(\Gr) = \sum_{v \in \Vertices} w(v)$.

A natural starting point for approximation algorithms for \VC is to solve the associated \LP:
\begin{align}
    \eqlab{lp:v:c}%
    &\text{min}
    &&
    {\textstyle\sum_{v \in \Vertices} w(v)x_v}
    \tag{\Term{L{}PVC}}
    \\
    &\text{s.t.}
    && x_u + x_v \geq 1
    && \forall uv \in
    \EdgesX{\Gr}
    &&
    \nonumber
    \\
    & && x_v \geq 0
    &\quad
    & \forall v \in \Vertices.
    \nonumber
\end{align}
For an assignment $\zz = (z_v)_{v \in \Vertices}$, let $\valX{\zz} = \sum_{v \in \Vertices} w(v) z_v$ denote the \emphi{value} of $\zz$.

\newcommand{\LPVCref}{\eqrefY{\LPVC}{lp:v:c}\xspace}
\begin{lemma}
    \lemlab{half}%
    There is an optimal solution for \emph{\LPVCref} that is half-integral. That is, for any $v \in \Vertices$, we have $x_v \in \{ 0,1/2,1\}$.
\end{lemma}
\begin{proof}
    Consider an optimal solution $\xx = (x_v)_{v\in\Vertices}$ for \LPVCref, Let
    \begin{equation*}
        L = \Set{ v \in \Vertices}{x_v \in (0,1/2)}%
        \qquad\text{and}\qquad%
        H = \Set{ v \in \Vertices}{x_v \in (1/2,1)}.
    \end{equation*}
    If $w(L) < w(H)$, then consider the assignment
    \begin{equation*}
        y_v =
        \begin{cases}
          x_v - \eps & v \in H\\
          x_v + \eps & v \in L \\
          x_v & \text{otherwise},
        \end{cases}
    \end{equation*}
    where $\eps = \min_{ v \in L \cup H} |x_v - 1/2|$. For an edge $uv$, with $u \in L$ and $v \in H$, we have $y_u + y_v = x_u +\eps +x_v - \eps = x_u + x_v \geq 1$. It is easy to verify, in a similar fashion, that $\yy = (y_v)_{v\in \Vertices}$ is a feasible solution for \LPVCref. Furthermore, since $w(L) < w(H)$, we have
    \begin{equation*}
        \valX{\yy}%
        =%
        \sum_{v \in \Vertices }  w(v) x_v - \eps w(H) + \eps w(L) %
        =
        \valX{\xx} + \eps( w(L) -  w(H)) %
        <%
        \valX{\xx},
    \end{equation*}
    which is impossible, by the optimality of $\xx$. A similar argument applies if $w(L) > w(H)$. As such, it must be that $w(L) = w(H)$. But then, the solution $\yy$ is also optimal, and it has one more variable assigned a value of $1/2$ than $\xx$. Repeating this argument now to $\yy$ till both $L$ and $R$ are empty implies the claim.
\end{proof}

\begin{lemma}
    An optimal half-integral solution to the linear program \emph{\LPVCref} can be computed in $O(n^3)$ time, where $n = \cardin{\Vertices}$.
\end{lemma}
\begin{proof}
    Let $\GrB$ be the bipartite graph over the bipartition $U_1 = \Set{u_1}{u \in \Vertices}$ and $U_2 = \Set{u_2}{u \in \Vertices}$. Let
    \begin{equation*}
        \EdgesX{\GrB} = \Set{u_1v_2,\, u_2 v_1}{ uv \in \EdgesX{\Gr}}.
    \end{equation*}
    For any vertex $v_i \in U_1 \cup U_2$, its weight is its original weight $w(v)$.

    There is an associated network flow instance -- adding a source vertex $s$, and a sink vertex $t$.  Here, the source vertex $s$ is connected to all the vertices in $U_1$. An edge $\dirEdgeY{s}{u_1}$, for $u_1 \in U_1$, has capacity $w(u)$. All the edges of $K$ are oriented from $U_1$ to $U_2$, with infinite capacity, and an edge $u_2 t$, for all $u_2 \in U_2$, has capacity $w(u_2) = w(u)$. Let $\GrC$ denote the resulting instance of network flow. It is easy to verify that an $s$-$t$ min-cut in $\GrC$ corresponds to a minimum weight vertex cover in $\GrB$, and vice versa.  As such, compute a max-flow in $\GrC$, and let $f$ denote this flow. Given $f$, the corresponding $s$-$t$ min-cut can be computed in linear time from the residual network flow of $f$ (i.e., min-cut max-flow theorem). This cut can then be readily converted into the desired minimum vertex cover $\Sopt$ in $\GrB$. Observe that $w(\Sopt) = |f|$, where $|f|$ is the value of the flow of $f$. The max-flow computation in this case can be done in $O(n^3)$ time using known algorithms \cite{clrs-ia-01}.

    For a vertex $v \in \VX{\GrB}$, let $\chi(v) =1$ if $v\in \Sopt$, and zero otherwise.  We set $x_v = \bigl( \chi(v_1) + \chi(v_2 ) \bigr) / 2$, for all $v \in \Vertices$.  Let $\xx = (x_v)_{v \in \Vertices}$.  We have
    \begin{equation}
        \valX{\xx}%
        =%
        \sum_{v \in \Vertices} w(v) x_v = \frac{w(\Sopt)}{2}.
        \eqlab{opt:v:c:o}
    \end{equation}
    Consider an edge $uv \in \EdgesX{\Gr}$. Observe that
    \begin{align*}
        x_u + x_v%
        &=%
        \frac{ \chi(u_1) + \chi(u_2) + \chi(v_1) +
           \chi(v_2)}{2}
        =%
        \frac{\chi(u_1) +  \chi(v_2)}{2}
        +
        \frac{\chi(u_2 ) + \chi(v_1 )}{2}
        \\&%
        \geq%
        \frac{1}{2} + \frac{1}{2} = 1,
    \end{align*}
    since $\Sopt$ is a vertex cover of $\GrB$. So $\xx = (x_v)_{v \in \Vertices}$ is a feasible assignment for \LPVCref (and it is also half-integral).

    Consider an optimal half-integral assignment $\yy = (y_v)_{v \in \Vertices}$ for \Term{L{P}VC}, which exists by \lemref{half}. This assignment induces a natural vertex cover for $\GrB$. Indeed, if $x_v = 1/2$ we add $v_1$ to a set $T$. Similarly, if $x_v = 1$, we add $v_1$ and $v_2$ to $T$. We claim that $T$ is a vertex cover for $\GrB$. Indeed, for any edge $uv \in \EdgesX{\Gr}$, we have that $x_u + x_v \geq 1$.  If $x_u=1$ then $u_1,u_2 \in T$ and these two vertices cover the edges $u_1v_2$ and $u_2v_1$. A similar argument applies if $x_v=1$. The remaining possibility is that $x_u = x_v =1/2$, but then $u_1, v_1 \in T$. Which implies that $T$ covers the edges $u_1v_2$ and $u_2v_1$. As such, all the edges of $\GrB$ are covered by $T$.

    Observe that
    \begin{equation*}
        2 \valX{\yy}%
        =%
        2 \sum_{v \in \Vertices} w(v)y_v %
        =%
        w(T)
        \geq%
        w(\Sopt) = 2 \valX{\xx},
    \end{equation*}
    by \Eqref{opt:v:c:o}.  Namely $\valX{\yy} \geq \valX{\xx}$.  We conclude that $\xx$ is an optimal solution for \LPVCref, and it is also half-integral.~
\end{proof}

Given an optimal half-integral solution $\xx = (x_v)_{v \in \Vertices}$ for \LPVCref, one can partition the vertices into three sets:
\begin{compactenumi}
    \smallskip%
    \item $V_0 = \Set{ v \in \Vertices }{ x_v =0}$,

    \smallskip%
    \item $V_{1/2} = \Set{ v \in \Vertices }{ x_v = 1/2}$, and

    \smallskip%
    \item $V_{1} = \Set{ v \in \Vertices }{ x_v =1}$.
\end{compactenumi}
\smallskip%
A few observations about these sets. The set $V_0$ is an independent set in $\Gr$. Indeed, an edge $uv \in \Edges$ with an endpoint $u \in V_0$ and an endpoint in $v \in V_0 \cup V_{1/2}$ has $x_u + x_v \leq 0 + 1/2 = 1/2$, which is impossible. As such, for a vertex $v \in V_0$, any adjacent edge $vu \in \Edges$ must have $u \in V_1$. Namely, there is no edge between a vertex of $V_0$ and a vertex $V_{1/2}$.  This is known as a \emphi{crown decomposition}, with $C=V_0$ being the ``crown'', $R= V_{1/2}$ being the ``body'', and $H = V_1$ be the ``head''.

We need the weighted version of the Nemhauser-Trotter theorem, which we prove next.

\begin{theorem}[Nemhauser-Trotter] %
    \thmlab{nt}%
    Let $\Gr$ be a graph with weights $w: \Vertices \rightarrow (0,\infty)$.  There is a minimum weight vertex cover $S$ of $\Gr$, such that $V_1 \subseteq S \subseteq V_1 \cup V_{1/2}$.
\end{theorem}

\begin{proof}
    Let $\Sopt$ be the optimal vertex cover, and let $S = (\Sopt \setminus V_0) \cup V_1$. Clearly, $S$ is a vertex cover, as all the edges adjacent to $V_0$ are covered by $V_1$. Since $V_0 \cup V_{1/2} \cup V_1 = \Vertices$, we have $S = V_1 \cup S \subseteq V_{1/2} \cup V_1$.

    Assume, for the sake of contradiction, that $w(S) > w(\Sopt)$. We have
    \begin{align*}
        &w(S)%
        =%
        w(\Sopt) - w\bigl(V_0 \cap \Sopt\bigr) +
        w\bigl(V_1 \setminus \Sopt\bigr) > w(\Sopt)\\
        &\iff
        w(V_0 \cap \Sopt) - w(V_1 \setminus \Sopt) < 0.
    \end{align*}
    Let $ \eps < 1/2$ be some constant, and for any $v\in \Vertices$, let
    \begin{equation*}
        y_v = \begin{cases}
          1 - \eps & v \in V_1 \setminus \Sopt\\
          \eps & v \in V_0 \cap \Sopt\\
          x_v & \mathrm{otherwise.}
        \end{cases}
    \end{equation*}
    An easy case analysis shows that the $\yy = (y_v)_{v \in \Vertices}$ is a feasible solution. Indeed, we need to consider only edges $uv$ with at least one endpoint (say $u$) in $V_1 \setminus \Sopt$, (as these are the only edges that lost value). There are the following cases to consider:
    \begin{compactenumi}
        \smallskip%
        \item $v \in V_1$: We have $y_u + y_v \geq 1 -\eps + 1 -\eps \geq 1$.

        \smallskip%
        \item $v \in V_{1/2}$: We have $y_u + y_v = 1 - \eps + 1/2 \geq 1$.

        \smallskip%
        \item $v \in V_0\cap \Sopt$: We have $y_u + y_v = 1 -\eps +\eps = 1$.

        \smallskip%
        \item $v \in V_0\setminus \Sopt$: Both endpoints of $uv$ are outside $\Sopt$, which is impossible.
    \end{compactenumi}
    \smallskip%
    Observe that $\valX{\yy} = \sum_{v \in \Vertices} y_v = \valX{\xx} + \eps w\pth{V_0 \cap \Sopt } - \eps w\pth{V_1 \setminus \Sopt} < \valX{\xx}$, which is a contradiction to the optimality of $\xx$.
\end{proof}

\begin{lemma}
    \lemlab{reduce}%
    Let $\GrB = \Gr[V_{1/2}]$ be the induced subgraph of $\Gr$ over $V_{1/2}$. Consider any minimum vertex cover $C$ of $\Gr[V_{1/2}]$. Then $C \cup V_1$ is a minimum weight vertex cover of $\Gr$.
\end{lemma}
\begin{proof}
    Let $\GrA= \Gr[V_0 \cup V_{1/2}]$.  Consider any minimum vertex cover $C$ to the graph $\GrA$ and observe that $C \cup V_1$ is a vertex cover for $\Gr$. Similarly, given a minimum vertex cover $\Sopt$ for $\Gr$, consider the equally priced vertex cover $S = (\Sopt \setminus V_0) \cup V_1$ (see \thmref{nt}), and observe that $S$ is a vertex cover of $\Gr$ and $S \setminus V_1$ is a vertex cover for $\GrA$.  We conclude that computing the minimum vertex cover for $\Gr$ is equivalent to computing the vertex cover for $\GrA$. Since $V_0$ is an independent set in $\Gr$, and it is connected to only the vertices of $V_1$, it follows that $V_0$ is a set of isolated vertices in $\GrA$. As such, it is sufficient to compute a minimum vertex cover for $\Gr[V_{1/2}]$. %
\end{proof}

Observe that $\xx$ induces a valid optimal fractional solution for the induced subgraph $\Gr[V_{1/2}]$.

\begin{lemma}
    Consider any fractional optimal solution $\yy$ for the vertex cover of the graph $\Gr[V_{1/2}]$. We have that $\valX{\yy} = \valX{\xx} - w(V_1) = \sum_{v \in V_{1/2}} w(v) /2 = w(V_{1/2})/2$.
\end{lemma}
\begin{proof}
    The assignment $\yy$, can be extended into a fractional solution for $\Gr$, by setting $y_v = 1$ if $v \in V_1$, and zero otherwise. Similarly, $\xx$ when restricted to $V_{1/2}$ is a valid \VC solution for $\Gr[V_{1/2}]$. This readily implies the claim, as $\xx$ is the optimal solution for \LPVCref.
\end{proof}

\section{The result and some applications}
\seclab{result}

For a weighted graph $\Gr$, let $\vcOpt(\Gr)$ denote the weight of the minimum weight vertex cover of $\Gr$. Similarly, let $\isOpt(\Gr)$ be the weight of the maximum weight independent set in $\Gr$.

\begin{theorem}
    \thmlab{main}%
    Let $\Gr=(\Vertices,\Edges)$ be a graph, with weights on the vertices, and assume that we are given an algorithm \alg that can compute, in $T(n)$ time, an independent set $I \subseteq \Vertices$, such that $w(I) \geq (1-\eps) \isOpt(\Gr)$. Furthermore, assume that this algorithm works for any induced subgraph of $\Gr$. Then, one can compute, in $O\bigl(n^3 + T(n)\bigr)$ time, a vertex cover $C$ of $\Gr$, such that $w(C) \leq (1+\eps)\vcOpt(\Gr)$.
\end{theorem}

\begin{proof}
    Using the algorithm of \lemref{reduce}, compute a partition of the vertices of $\Gr$ into the three sets $V_0, V_{1/2}, V_1$.  This partition reduces $\Gr$ into an induced subgraph $\GrB = \Gr[V_{1/2}]$, such that it is enough to solve the problem on $\GrB$. Importantly, for the graph $\GrB$, we have the property that the optimal \LP solution (for the relaxation of the \LP) assigns all vertices a value $1/2$. That is $\vcOpt = \vcOpt( \GrB) \geq w/2$, where $w = \sum_{v \in \VX{\GrB} } w(v)$.

    In the following, let $\isOpt = \isOpt(\GrB)$.  Compute an independent set $I$ in $\GrB$ using \alg.  We have that
    \begin{equation*}
        w(I) \geq (1-\eps) \isOpt
        =%
        (1-\eps)(w  - \vcOpt).
    \end{equation*}
    The set $C = \VX{\GrB} - I$ is a vertex cover, and we have
    \begin{align*}
        w(C)%
        &=%
        w - w(I)%
        \leq%
        w
        - (1-\eps)(w - \vcOpt)
        =%
        \eps w + \pth{1 - \eps} \vcOpt
        \\&%
        \leq%
        2 \eps \, \vcOpt + \pth{1 - \eps} \vcOpt
        =%
        (1+\eps) \vcOpt,
    \end{align*}
    since $\vcOpt \geq w/2$. We have that $C \cup V_1$ is a vertex cover for the original graph. Furthermore, by \lemref{reduce}, we have
    \begin{align*}
        w(C \cup V_1)%
        &=%
        w(C) + w(V_1)%
        \leq%
        (1+\eps) \vcOpt
        + w(V_1)
        \\&
        \leq%
        (1+\eps) (w(V_1) + \vcOpt)
        =%
        (1+\eps)  \vcOpt(\Gr),
    \end{align*}
    since $\vcOpt(\Gr) = \vcOpt + w(V_1)$ by \lemref{reduce}.
\end{proof}

\subsection{Applications}

For several cases, efficient $(1+\eps)$-approximations algorithms are known:
\begin{compactenumi}
    \smallskip
    \item A \PTAS for the unweighted pseudo-disks \cite{ch-aamis-12}.

    \smallskip
    \item A \QQPTAS is for unweighted rectangles in the plane \cite{ce-amisr-16}.

    \smallskip
    \item A \QPTAS for weighted simple polygons in the plane \cite{ahw-asiss-19}.
\end{compactenumi}
\smallskip%
Plugging these results into \thmref{main} implies the following.

\begin{theorem}
    We have the following:
    \begin{compactenumi}
        \smallskip%
        \item For the intersection graph of $n$ unweighted pseudo-disk in the plane, a $(1+\eps)$-approximation to the minimum vertex cover can be computed in $n^{O(1/\eps^2)}$ time.

        \smallskip%
        \item For the intersection graph of $n$ unweighted axis-aligned rectangles in the plane, a $(1+\eps)$-approximation to the minimum vertex cover can be computed in $n^{O(\poly(\log \log n, 1/\eps))}$ time.

        \smallskip%
        \item For the intersection graph of $n$ weighted simple polygons in the plane, a $(1+\eps)$-approximation to the minimum weight vertex-cover can be computed in $n^{O(\poly(\log n, 1/\eps))}$ time.
    \end{compactenumi}
\end{theorem}

\paragraph{Acknowledgments.} %
The author thanks Chandra Chekuri for useful discussions, and Timothy Chan and C. S. Karthik for pointing out relevant references.

\printbibliography%

\end{document}